%% file: main.tex
\documentclass[11pt]{article}
\usepackage[T1]{fontenc}
\usepackage{lmodern}
\usepackage{enumitem}

\usepackage{amsthm}
\usepackage{graphicx} 
\usepackage{array} 

\usepackage{amsmath, amssymb, amsfonts, verbatim}
\usepackage{hyphenat,epsfig,subfigure,multirow}
\usepackage{algorithmicx}
\usepackage{algpseudocode}

\usepackage[usenames,dvipsnames]{xcolor}

\usepackage{tcolorbox}

\definecolor{DarkRed}{rgb}{0.5,0.1,0.1}
\definecolor{DarkBlue}{rgb}{0.1,0.1,0.5}

\usepackage{hyperref}
\hypersetup{
colorlinks=true,
pdfnewwindow=true,
citecolor=ForestGreen,
linkcolor=DarkRed,
filecolor=DarkRed,
urlcolor=DarkBlue
}

\usepackage{bm}
\usepackage{url}
\usepackage{xspace}
\usepackage[mathscr]{euscript}

\usepackage{mdframed}

\usepackage{cite}
\usepackage{enumitem}

\usepackage[margin=1in]{geometry}

\newtheorem{mdinvariant}{Invariant}

\allowdisplaybreaks

\DeclareMathOperator*{\OV}{{\bf OV}}
\DeclareMathOperator*{\CP}{{\bf CP}}

\input{notations}

\title{Faster Algorithms for Average-Case Orthogonal Vectors and Closest Pair Problems}

\author{
Josh Alman\thanks{\texttt{josh@cs.columbia.edu}. Columbia University. Supported in part by NSF Grant CCF-2238221 and a grant from the Simons Foundation (Grant Number 825870 JA).}
\quad
Alexandr Andoni\thanks{\texttt{andoni@cs.columbia.edu}. Columbia University. Supported in part by NSF grant CCF2008733 and ONR grant N00014-22-1-2713.}
\quad
Hengjie Zhang\thanks{\texttt{hengjie.z@columbia.edu}. Columbia University. Supported in part by NSF grant CCF2008733 and ONR grant N00014-22-1-2713.}
}

\date{}
\begin{document}
\maketitle
\thispagestyle{empty}

\abstract{

We study the average-case version of the Orthogonal Vectors problem, in which one is given as input $n$ vectors from $\{0,1\}^d$ which are chosen randomly so that each coordinate is $1$ independently with probability $p$. Kane and Williams [ITCS 2019] showed how to solve this problem in time $O(n^{2 - \delta_p})$ for a constant $\delta_p > 0$ that depends only on $p$. However, it was previously unclear how to solve the problem faster in the hardest parameter regime where $p$ may depend on $d$.

The best prior algorithm was the best worst-case algorithm by Abboud, Williams and Yu [SODA 2014], which in dimension $d = c \cdot \log n$, solves the problem in time $n^{2 - \Omega(1/\log c)}$. In this paper, we give a new algorithm which improves this to $n^{2 - \Omega(\log\log c /\log c)}$ in the average case for any parameter $p$.

As in the prior work, our algorithm uses the polynomial method. We make use of a very simple polynomial over the reals, and use a new method to analyze its performance based on computing how its value degrades as the input vectors get farther from orthogonal.

To demonstrate the generality of our approach, we also solve the average-case version of the closest pair problem in the same running time.

}

\clearpage
\pagenumbering{arabic} 

\section{Introduction}

Fine-grained complexity theory aims to determine the precise running times of important algorithmic problems which have polynomial-time algorithms. One typically starts with a fine-grained hardness assumption, which says that the best-known running time for a key problem cannot be improved by a polynomial factor, and then one uses fine-grained reductions to transfer that hardness assumption to many other problems of interest. See, for instance, the survey~\cite{williams2018some} for a detailed overview and many examples.

\paragraph{Orthogonal Vectors}
The Orthogonal Vectors (OV) problem is at the heart of one of the most popular hardness assumptions. In this problem, one is given as input $n$ vectors from $\{0,1\}^d$, and the goal is to determine whether any pair of them is orthogonal (over $\mathbb{Z}$). The straightforward algorithm runs in $O(n^2 d)$ time, and there are two nontrivial algorithms known: A (folklore) divide-and-conquer algorithm running in time $O(n + d \cdot 2^d)$ for small $d\leq 2\log n$, and an algorithm by Abboud, Williams, and Yu~\cite{abboud2014more} based on the \emph{polynomial method} which runs in time $n^{2 - 1/O(\log c)}$ for dimension $d = c \log n$. (See also~\cite{chan2016deterministic} for a derandomization of the polynomial method algorithm.)

The OV Conjecture states that when $d = \Omega(\log n)$, this problem cannot be solved in truly-subquadratic time:
\begin{conjecture}[OV Conjecture]
    For every $\delta>0$, there exists a $c>0$ such that $OV$ for $d = c \log n$ cannot be solved in time $O(n^{2-\delta})$.
\end{conjecture}
\noindent By contrast, the algorithm of~\cite{abboud2014more} shows that for every $c>0$ there exists a $\delta = 1/O(\log c) > 0$ such that $OV$ for $d = c \log n$ can be solved in time $O(n^{2-\delta})$.

The OV Conjecture is implied by the Strong Exponential Time Hypothesis (SETH)~\cite{impagliazzo2001complexity,williams2005new}, and most SETH-hardness results in the literature actually proceed by reducing from OV. The OV Conjecture is known to prove hardness for a variety of problems in graph algorithms, string algorithms, dynamic algorithms, computational biology, nearest neighbor search, and other areas; see~\cite[Section 3]{williams2018some}.

\paragraph{Average-Case Fine-Grained Complexity}
A more recent line of work has begun studying \emph{average-case} fine-grained complexity~\cite{kane2017orthogonal,ball2017average,ball2018proofs,goldreich2018counting,lavigne2019public,boix2021average,dalirrooyfard2020new,brakerski2020hardness}. This area investigates whether problems like OV, which are assumed to be hard in the worst-case, remain hard on random inputs. As with (non-fine-grained) average-case complexity, there are two primary motivations for this line of work.

The first motivation is to solving instances of problems which arise in practice more quickly. For instance, if OV can be solved faster on average, then even if a problem $P$ is known to be hard assuming the OV Conjecture, it is possible that realistic inputs to $P$ may be drawn as in the average case and solved faster as well.

The second motivation is to develop a foundation for \emph{fine-grained cryptography}. Usual cryptographic primitives stipulate that an adversary cannot compromise their security in polynomial time; in fine-grained cryptography, we weaken the requirement and only require that an adversary cannot compromise the security in a specific polynomial time (like quadratic or cubic). This weaker requirement may suffice in many situations, and it may be achievable even if our usual cryptographic assumptions are false. The bedrock of such a theory needs to be an average-case fine-grained hard problem.

\paragraph{Average-Case OV}
Toward these goals, in this paper we study the average-case version of the OV problem, in which the entries of the input vectors are drawn independently and identically.

\begin{definition}[$\OV(p)_{n,d}$]\label{def:ov}
We are given as input two sets $X,Y$ of $n$ vectors each from $\{0,1\}^d$, where each vector is chosen i.i.d. as follows: every bit is set to be $1$ with probability $p$ and $0$ otherwise. In the $\OV(p)_{n,d}$ problem, with probability at least $1 - o(1)$, we must output an orthogonal pair $(x,y) \in X \times Y$ with $\langle x,y \rangle = 0$ if one exists, and determine that there is no orthogonal pair otherwise.
\end{definition}

This problem was first studied by Kane and Williams~\cite{kane2017orthogonal}, who showed that for every $p>0$ there is a $\delta_p = \Theta(p) > 0$ such that $\OV(p)_{n,d}$ can be solved in time $O(n^{2 - \delta_p})$. In fact,~\cite{kane2017orthogonal} showed that even $\mathsf{AC^0}$ branching programs, a very simple model of computation, are able to solve $\OV(p)_{n,d}$ with this size. 
Subsequent work by Dalirrooyfard, Lincoln and Vassilevska Williams~\cite{dalirrooyfard2020new} showed that one can similarly \emph{count} the number of orthogonal pairs in instances drawn from the $\OV(p)_{n,d}$ distribution in time $O(n^{2-\delta'_p})$ for some $\delta'_p = \Theta(\frac{p^4}{\log(1/p)}) > 0$.

Kane and Williams' algorithm solves $\OV(p)_{n,d}$ in truly subquadratic time only when one thinks of $p$ as a fixed constant. 
However, as they observe (and we use here as well), the only non-trivial regime of $p$ is around $\sqrt{\frac{2\ln 2}{c}}$, for $c=d/\log n$, which is when the expected number of orthogonal pairs is a constant.\footnote{For larger values $p$, there is $o(1)$ probability that there exists an orthogonal pair, so the algorithm doesn't need to do anything. For $p$ much smaller, there are so many orthogonal pairs that we can subsample the input set and just run the straight-forward algorithm in $n^{2-\Omega(1)}$ time.}

Indeed, in this hardest setting where $p = 1/\Theta(\sqrt{c})$, Kane and Williams' algorithm runs in time $n^{2 - 1/O(\sqrt{c})}$, which is actually slower than the worst-case algorithm~\cite{abboud2014more} which runs in time $n^{2 - 1/O(\log c)}$. The ideas from the prior algorithms~\cite{kane2017orthogonal,dalirrooyfard2020new} do not seem to yield any improvements to the worst-case algorithm in this setting, since they rely on arguments about the sparsities of vectors which become insignificant as $p \to 0$. The worst-case algorithm~\cite{abboud2014more} also does not run faster in this average-case setting, since it relies on arguments about \emph{parities} of inner products of random sub-vectors, which are not substantially altered by sparsity or randomness.

In other words, for the hardest choice of $p$, the $n^{2 - 1/O(\log c)}$ worst-case running time seems to be the best we can achieve in the average case $\OV(p)_{n,d}$ as well.

\paragraph{New Faster Algorithm}
Our main result is a new, faster algorithm for $\OV(p)_{n,d}$, no matter how small $p$ is.

\begin{theorem}\label{thm:main}
For any fixed constant $c>1$, and $d=c\log n$, and any $p\in(0,1)$, we can solve $\OV(p)_{n,d}$ in $n^{2-\Omega(\frac{\log\log c}{\log c})}$ time.
\end{theorem}

We improve the savings in the exponent from $\Omega(\frac{1}{\log c})$ to $\Omega(\frac{\log\log c}{\log c})$. To our knowledge, this is the first algorithm for average-case OV that is faster than the best worst-case OV algorithm in all parameter regimes. Although this is not a truly subquadratic-time algorithm, it gives evidence that the average case may indeed be easier than the worst case. 

Our algorithm uses a new variant on the polynomial method, the same high-level technique which is used by~\cite{abboud2014more}. Whereas the prior algorithm works over the finite field $\mathbb{F}_2$ in order to save low-order factors in the degree they need, we instead use a polynomial over $\mathbb{R}$ which we carefully select to have good performance in expectation. Most importantly, we leverage the fact that our polynomial has a smooth degradation between pairs that are orthogonal and pairs that are far from orthogonal (which are most common).



\paragraph{Average-Case Closest Pair}

In order to demonstrate the generality of our approach, we also use it to solve the average-case version of the closest pair problem.

\begin{definition}[$\CP_{n,d}$]\label{def:cp}
We are given as input two sets $X,Y$ of $n$ vectors each from $\{0,1\}^d$, where each vector is chosen uniformly and independently at random. In the $\CP_{n,d}$ problem, with probability at least $1 - o(1)$, we must output a pair $(x,y) \in X \times Y$ with minimum Hamming distance.
\end{definition}

This is a problem where sparsity does not come into play: all the vectors have roughly half 0s and half 1s. By contrast, the hard case of $\OV(p)_{n,d}$ is when $p$ is very small, a fact we take advantage of in our algorithm. Nonetheless, our technique leads to a new faster algorithm for $\CP_{n,d}$ as well: The fastest known algorithm in the worst case by Alman, Chan and Williams~\cite{alman2015probabilistic,alman2016polynomial}, which also uses the polynomial method, runs in time $n^{2 - \tilde{\Omega}(1/\sqrt{c})}$, and our new average-case algorithm improves this to $n^{2 - \Omega(\frac{\log\log c}{\log c})}$:

\begin{theorem}\label{thm:main2}
For any fixed constant $c>1$, and $d=c\log n$, we can solve $\CP_{n,d}$ in $n^{2-\Omega(\frac{\log\log c}{\log c})}$ time.
\end{theorem}

\noindent This is an even larger saving in the exponent than before, from $\tilde{\Omega}(1/\sqrt{c})$ to $\Omega(\frac{\log\log c}{\log c})$.

We briefly compare $\CP_{n,d}$ with the \emph{light bulb} problem, another popular variant on the closest pair problem with random inputs. In the light bulb problem, one is given as input $n$ vectors from $\{0,1\}^d$ each chosen uniformly and independently, as well as a \emph{planted} closest pair of vectors which has Hamming distance at most $(1/2 - \rho)d$ for some constant $\rho>0$, and the goal is to find the planted pair. In other words, the light bulb problem is a ``planted version'' of $\CP_{n,d}$, where the planted pair is much closer than the closest pair will be in $\CP_{n,d}$. By taking advantage of how close the planted pair is, one can solve the light bulb problem in truly subquadratic time $O(n^{1.62})$~\cite{valiant2015finding,karppa2018faster,karppa2016explicit,alman2018illuminating,alman2023generalizations}. This running time can even be achieved using the polynomial method~\cite{alman2018illuminating}. Here we focus instead on the harder $\CP_{n,d}$.

\subsection{Technique Overview}

\paragraph{Prior Average-Case Algorithm}

Kane and Williams' average-case OV algorithm~\cite{kane2017orthogonal} critically makes two observations about the sparsities of vectors in $\OV(p)_{n,d}$. The first, as mentioned before, is that, for $d = c \log n$, we need only consider $p$ that is around $\sqrt{\frac{2\ln 2}{c}}$, which is when we have a constant number of orthogonal pairs in expectation. 

The second is that, if we draw vectors $u,v \in \{0,1\}^d$ where each coordinate is $1$ with probability $p$, but then condition on the event that $u$ and $v$ are orthogonal, then the expected number of $1$s in $u$ or $v$ is only $d p/(1+p)$. Thus, to solve $\OV(p)_{n,d}$, we may focus on vectors with sparsity close to $d p/(1+p)$. However, there are only $n^{1 - \Omega(p)}$ such vectors with high probability, so the trivial quadratic time algorithm applied to these vectors runs in truly subquadratic time. It seems difficult to further improve this algorithm without new ideas, since there really are likely to be $n^{1 - O(p)}$ vectors as sparse as the orthogonal pair.

\paragraph{Prior Worst-Case Algorithm}

Abboud, Williams and Yu's worst-case OV algorithm~\cite{abboud2014more} uses the polynomial method. This method, at a high level, consists of two main components:
\begin{enumerate}
    \item Design a low-degree polynomial which determines whether two input vectors are orthogonal, then
    \item Use a now-standard technique (based on fast rectangular matrix multiplication) to quickly evaluate the polynomial on many pairs of inputs and find ones which the polynomial accepts.
\end{enumerate}
In our new algorithm in this paper, we will also use this approach, and we will use batch polynomial evaluation in the same way here, but we will design a different polynomial. 

The polynomial of \cite{abboud2014more} is a probabilistic polynomial over $\mathbb{F}_2$, i.e., a randomly drawn polynomial $p : \{0,1\}^d \times \{0,1\}^d \to \{0,1\}$ such that, if $u,v \in \{0,1\}^d$ are orthogonal then $p(u,v) = 1$, and if $u,v$ are not orthogonal then $\Pr[p(u,v) =0] \geq 1-\varepsilon$. They use a classic polynomial construction of Razborov and Smolensky~\cite{razborov1987lower,smolensky1987algebraic} from circuit complexity theory, which achieves degree $O(\log 1/\varepsilon)$ over $\mathbb{F}_2$. Meanwhile, the best-known probabilistic polynomial over $\mathbb{R}$ has worse degree $O(\log d \log 1/\varepsilon)$~\cite{tarui1993probabilistic,beigel1990perceptron,bhandari2021probabilistic}.

\paragraph{Our New Algorithm} Our new algorithm instead uses a simple, deterministic polynomial over $\mathbb{R}$, for an appropriate choice of degree $q \in \mathbb{N}$:
$$P(u,v) = (p^2d - \langle u,v\rangle)^q.$$
Intuitively, $P(u,v)$ is large when $u,v$ are orthogonal (in that case $P(u,v) = (p^2 d)^q$), but $P(u,v)$ is small for a typical random $u,v$ (since the expected value of $\langle u,v \rangle$ is $p^2 d$). Thus, when we add together many evaluations of $P$, the result will be large if any of the evaluation pairs are orthogonal, and small otherwise.

If we argue that $|P(u,v)| \leq T$ for a typical random $u,v$, then it would follow that we can add together $(p^2 d)^q/T$ evaluations to detect an orthogonal pair. This type of ``threshold'' argument is prevalent in prior algorithmic uses of the polynomial method over $\mathbb{R}$ (e.g.,~\cite{chan2016deterministic,alman2016polynomial,chen2018hardness}, including the previous work on the light bulb problem~\cite{alman2018illuminating}). We instead argue about the \emph{expected value} of $P(u,v)$ over randomly drawn $u,v$ and find that it is much smaller than $T$, ultimately leading to our improved bound. Put differently, in contrast to the prior uses, rather than caring about only two values of the polynomial (for orthogonal vs non-orthogonal pair), we crucially use different bounds on the polynomial value for different regimes (e.g., far from orthogonal and closer to orthogonal). 
Rather than delving into more details here, we direct the reader to the full, short proof below.


\section{Orthogonal Vectors under Random distribution}
\subsection{Preliminary}
In this section, we give our new algorithm for $\OV(p)_{n,d}$, with $d = c \log n$. 
As mentioned previously, we will assume $\sqrt{1/c} \leq p \leq 3 \sqrt{1/c}$. This is because the problem is straightforward otherwise: We can calculate that if $x,y \in \{0,1\}^d$ are drawn with each entry $1$ with probability $p$ i.i.d., then 
\[
\Pr_{x,y}[\langle x,y \rangle = 0] = (1-p^2)^d = (1/e)^{(1+o(1))p^2d} = n^{-(1+o(1))\ln 2 \cdot p^2c}.
\]
If $p < \sqrt{1/c}$, there will be more than $n^{0.1}$ orthogonal pairs with high probability, and we can subsample and solve the problem on a sublinear number of vectors. If $p > 3 \sqrt{1/c}$, there will be no orthogonal pair with high probability, so we can immediately return.

Our algorithm and analysis will make use of two key tools. First is the standard Chernoff bound.

\begin{lemma}[The Chernoff bound]\label{lem:chernoff} Given $n$ i.i.d. random variables $x_1,\ldots,x_n$. Let $X=\sum x_i$ and $\mu = E[X]$, then, for $0< \delta < 1$, we have
\[
    \Pr[|X-\mu|\geq \delta\cdot \mu]\leq 2e^{-\delta^2\mu/3}.
\]
\end{lemma}

Second is a tool for quickly batch-evaluating a polynomial, which is a core technique in most uses of the polynomial method in algorithm design.

\begin{lemma}[Fast Polynomial Evaluation~\cite{williams2014faster}]\label{lem:fast_rect}
Given a $2D$-variate polynomial $P(x_1, \ldots, x_D, y_1, \ldots, y_D)$ over the integers with $M \leq N^{0.1}$ monomials, and where each coefficient is a $\text{polylog}(N,D)$-bit integer, and given $A,B \subseteq \{0,1\}^D$ such that $|A|, |B| \leq N$, we can evaluate $P$ on all pairs $(a_i, b_j) \in A \times B$ in $\tilde{O}(N^2 + D \cdot N^{1.2})$ time.
\end{lemma}

Lemma~\ref{lem:fast_rect} is based on an algorithm by Coppersmith for rectangular matrix multiplication~\cite{coppersmith1982rapid}. We note that this algorithm, which uses a small algebraic identity similar to Strassen's algorithm, is not necessarily impractical~\cite{williams2014faster}.

\subsection{Our algorithm and analysis}
Given the definition of the OV problem from Def.~\ref{def:ov}, the inputs are two sets $X,Y \subset \{0,1\}^d$ for $d = c \log n$. We arbitrarily partition $X$ into groups $X_1,\cdots,X_{n/s}$ with equal size $s$, where $s = n^{\Theta(\frac{\log \log c}{\log c})}$ for a constant hidden by the $\Theta$ to be determined later. We define the degree $q$ single-variable polynomial $Q : \mathbb{R} \to \mathbb{R}$ as $Q(z):= (z-p^2 d)^q$ for a degree $q$ to be determined. We will pick $q$ to be even so that $Q(z)\geq 0$ for all inputs $z \in \mathbb{R}$. Note that $p^2 d$ is the expected inner product between two random vectors $x,y \in \{0,1\}^d$ drawn from our distribution. Our algorithm is as follows. 

\subsubsection{Algorithm.}
Simultaneously, for every group $X_i$ ($i \in [n/s]$) and vector $y_j \in Y$ ($j \in [n]$), we calculate the quantity $$A_{i,j} := \sum_{x\in X_i} Q(\langle x,y_j \rangle).$$ We do this using $s$ calls to Lemma~\ref{lem:fast_rect} with $N = n/s$ as follows:  Given a group $X_i$ ($i\in [n/s]$) of vectors $X_i = \{ x_1, \ldots, x_s\}$ and a vector $y_j$ ($j\in [n]$), we design the polynomial $P$ to use in Lemma~\ref{lem:fast_rect} as 
\[
P(x_{1,1},\ldots,x_{1,d},\ldots,x_{s,1}\ldots,x_{s,d},y_1,\ldots,y_d):= \sum_{\ell=1}^s Q(\langle x_\ell, y\rangle).
\]
Thus, $A_{i,j}$ is indeed equal to $P(\{x\}_{x\in X_i},y_j)$. We will apply Lemma~\ref{lem:fast_rect} $s$ times, each time with $N = n/s$. Each time we will set $A = \{X_i\}_{i \in [n/s]}$, and we will set $B$ to be $n/s$ out of the $y_j$s, so that across all $s$ calls, we will evaluate $P$ on all $X_i$ and all $y_j$.

Finally, once we've computed all the $A_{i,j}$ values, for each $i,j$ with $A_{i,j} \geq (p^2 d)^q$, we brute force check whether any $x\in X_i$ is orthogonal to $y_j$. 

It is not hard to verify the \emph{correctness} of this algorithm: If $x \in X_i$ and $y_j$ are orthogonal, then $A_{i,j} \geq Q(\langle x,y_j\rangle) = (\langle x,y_j\rangle-p^2 d)^q = (p^2 d)^q$. We will thus detect it with brute force.

It remains to prove that, with a proper choice of $q$ and $s$, this algorithm will run in time $n^{2-\Omega(\frac{\log \log c}{\log c})}$ with high probability.

\subsubsection{Running Time Analysis}\label{subsubsec:analysis}

The key behind the proof of the running time of our algorithm is the following lemma, which says that we are unlikely to need to use brute force very often. 

\begin{lemma}\label{lem:core_lemma}
Let $q = \frac{\log n}{50\log c}$. There's a choice of the group size $s$, $s = n^{\Theta(\frac{\log \log c}{\log c})}$, such that the following is true. Consider a group $X_i$ and a vector $y_j$. With at least $1 - n^{-1/20000}$ probability over the randomness on $X_i$ and $y_j$, we have 
\[
A_{i,j} < (p^2d)^q.
\]
\end{lemma}

We first prove that our algorithm has the desired running time assuming Lemma~\ref{lem:core_lemma} is true, then we will prove Lemma~\ref{lem:core_lemma} below.

We set $q = \frac{\log n}{50\log c}$. Since we calculate all $A_{i,j}$ using Lemma~\ref{lem:fast_rect}, let us verify that the conditions of that lemma are met. For $x,y$ vectors of length $d$, we can expand $Q(\langle x,y\rangle)$ into a sum of monomials, and since we will only be evaluating it on $x,y \in \{0,1\}^d$, we may further assume that $Q(\langle x,y\rangle)$ is expanded into a sum of \emph{multilinear} monomials. Thus, since it has degree $q$, its number of monomials will be at most $\sum_{\ell=0}^q \binom{2d}{\ell}$. Using the standard bound that $\binom{n}{k} \leq (n \cdot e / k)^k < (3 n / k)^k$, we can thus upper bound the number of monomials by $$\sum_{\ell=0}^q \binom{2d}{\ell} \leq \sum_{\ell=0}^q \binom{2c \log n}{\frac{\log n}{50\log c}} \leq 
q \cdot (300c \log c)^{\frac{\log n}{50 \log c}}
\leq n^{\frac{(1+o(1))\log c}{50\log c}}
\ll n^{0.05}$$
when $c$ is small enough. Finally, it is not hard to verify that all coefficients of $P$ are $\text{polylog}(n)$-bit integers.

Since $P$ is a sum of $s$ copies of $Q$, the total number of monomials in $P$ is $\ll n^{0.05} \cdot s < (n/s)^{0.1}$ since $s=n^{\Theta(\frac{\log \log c}{\log c})} \ll n^{0.01}$. Thus, since the running time of each call to Lemma~\ref{lem:fast_rect} is $N^{2 + o(1)}$ where $N=n/s$, it follows that the total running time for evaluating all the $A_{i,j}$ is $s \cdot N^{2 + o(1)} \leq n^{2+o(1)}/s = n^{2-\Omega(\frac{\log \log c}{\log c})}$. 

As for the time spent on brute force, recall that for each $i \in [n/s]$ and $j \in [n]$ we will run brute force if $A_{i,j} < (p^2d)^q$. We say the pair $(i,j)$ ``fails'' if $A_{i,j} < (p^2d)^q$. We now use Lemma~\ref{lem:core_lemma}. For notational convenience, define $\eps := 1/10000$. By picking $s = n^{\Theta(\frac{\log \log c}{\log c})}$, the probability that a given $(i,j)$ fails is only $n^{-0.5\eps}$ by Lemma~\ref{lem:core_lemma}. There are $n^2/s \leq n^2$ pairs of $(i,j)$, so the number of pairs that fail in expectation is $\leq n^{2-0.5\eps}$. Thus, by Markov's inequality, with probability at most $p':=n^{-0.1\eps}$, there are at most $n^{2-0.5\eps}/p'=n^{2-0.4\eps}$ pairs that fail. We spend $O(sd)$ time in the brute force for every failed pair. Since $s\leq n^{0.1\eps}$ for large enough $c$, this means that in total, we spend $n^{2-0.4\eps}sd\leq n^{2-0.3\eps}$ time in brute force, with probability $1-o(1)$. This is truly-subquadratic time, much smaller than our desired $n^{2-\Omega(\frac{\log \log c}{\log c})}$.

\subsubsection{Key Analysis Lemma}

We finally conclude our proof by proving Lemma~\ref{lem:core_lemma}. Again, for convenience, we let $\eps$ be the small constant $\eps:=1/10000$.
Our proof will proceed by defining three random events, $E_1$, $E_2$ and $E_3$, then proving that
\begin{enumerate}[label=\textbf{(\arabic*)}]
    \item The probability that all three of $E_1$, $E_2$ and $E_3$ happen is at least $1-n^{-\eps/2}$, and 
    \item If all three happen, then $A_{i,j} < (p^2d)^q$.
\end{enumerate}
These together will conclude the proof. 
We now define the events. Roughly, they each say that the densities and inner products of the randomly drawn vectors are close to their expectations.

\begin{itemize}
\item Let $E_1$ be: $\|y_j\|_1 \in pd \pm 0.1 c^{1/4} \log n$.
\item Let $E_2$ be: for all $x\in X_i$, we have $\langle x,y_j\rangle \in p^2 d \pm 0.2 \log n$. 
\item Let $E_3$ be the following. For all $\alpha$ in $[-0.2, 0.2]$ for which $\alpha \log n$ is an integer, let $n_{\alpha}$ be the number of $x\in X_i$ for which $\langle x,y_j\rangle = p^2 d + \alpha \log n$. We define $E_3$ to be the event that
\[
n_{\alpha} \leq \log n \cdot (1 + s \cdot n^{-\eps\alpha^2}), ~~~~~~\forall \alpha \in [-0.2,\frac{0.2}{c^{1/4}}] \cup [\frac{0.2}{c^{1/4}},0.2].
\]
\end{itemize}

\noindent\textbf{Proof of (1).} 
We will analyze each of the events individually and then apply a union bound.

\begin{itemize}
\item We have $\Pr[\neg E_1] \leq n^{-\eps}$, by using the Chernoff bound, Lemma~\ref{lem:chernoff}, where $\mu = pd$ is between $\sqrt{c}\log n$ and $3\sqrt{c}\log n$, and $\delta = \frac{0.1c^{1/4}\log n}{pd} \geq 0.03 c^{-1/4}\log n$. 
\item We have $\Pr[\neg E_2] \leq s\cdot n^{-\eps}$. Indeed, again, by the Chernoff bound, we have for any $x,y$ drawn from the distribution that
    \[
    \Pr_{x,y}[|\langle x,y \rangle - p^2d| \geq 0.2\log n] \leq n^{-0.01}.
    \]

    By a union bound on all $s$ choices of $x\in X_i$, we get that $\Pr[\neg E_2] \leq s\cdot n^{-0.01}\leq s\cdot n^{-\eps}$.
    
\item We have $\Pr[\neg E_3 | E_1]\leq 1/n^{10}$. To prove this, we first fix $y_j \in \{0,1\}^d$. Conditioned on $E_1$ being true, we have $\|y_j\|_1 \in pd \pm 0.1 c^{1/4}\log n$. Now, for a random vector $x$, the distribution of their inner product $\langle x, y_j\rangle$ is equivalent to the summation of $\|y_j\|_1$ iid random variables. Each equals $1$ with probability $p$ and $0$ otherwise. The expectation is thus $E[\langle x, y_j\rangle] = p\|y_j\|_1 \in p^2d\pm \frac{0.5}{c^{1/4}}\log n$. 

Now fix any $\alpha \in [\frac{1}{c^{1/4}},0.2]$, so that we have $(p^2 d + \alpha \log n) - E[\langle x,y_j\rangle] \geq 0.5 \alpha\log n$. So by a Chernoff bound with $\delta = \frac{0.5\alpha \log n}{2p^d}\geq 0.02\alpha$, we have
\[
\Pr[\langle x,y_j\rangle \geq p^2 d + \alpha \log n]\leq n^{- \alpha^2/10000} = n^{- \eps \alpha^2}.
\]
Thus, each $x\in X_i$ contributes $1$ to $n_{\alpha}$ with probability $\leq n^{-\eps \alpha^2}$, so $E[n_{\alpha}]\leq |X_i|n^{-\eps\alpha^2} = sn^{-\eps\alpha^2}$. Since the vectors $x\in X_i$ are drawn independently from each other, we can use the Chernoff bound to bound the probability that $n_{\alpha}$ is much larger than its expectation. In the case that $E[n_{\alpha}] \leq 1$, we will have $n_{\alpha} \leq \log n$ with high probability. When $E[n_{\alpha}]>1$, with high probability $n_{\alpha} \leq (\log n)E[n_{\alpha}]$.
Thus in total, we have $n_{\alpha} \leq s\log n\cdot (1+n^{-\eps\alpha^2})$ with probability $1-1/n^{10}$.

For $\alpha \in [-0.2 ,-\frac{1}{c^{1/4}}]$, an identical analysis applies.
\end{itemize}

By a union bound, we have as desired that for big enough $c$ (and thus small enough $s$),
\[
\Pr[E_1,E_2,E_3]\geq 1-n^{-\eps} - sn^{-\eps}-1/n^{10} \geq 1-n^{-\eps/2}.
\]

\noindent\textbf{Proof of (2)}

Now we will choose the group size $s$ so that, when all three of $E_1$, $E_2$ and $E_3$  happen, then we always have $A_{i,j} < (p^2d)^q$.

Recall that $A_{i,j} = \sum_{x\in X_i} Q(\langle x,y_j \rangle) = \sum_{x\in X_i} (\langle x,y_j \rangle - p^2d)^q$, and recall the definition of $n_\alpha$ from $E_3$. Supposing $E_2$ is true, we can rewrite $A_{i,j}$ in another form, that enumerates the number of $x$ with $\langle x,y_j\rangle = p^2d + \alpha \log n$ for each $\alpha$, i.e., $A_{i,j} = \sum_{\alpha \in [-0.2,0.2]} n_{\alpha} (\alpha \log n)^q$. We will break the choices of $\alpha$ into two cases and bound their contributions to the sum separately:
\[
A_{i,j} = \underbrace{\sum_{\alpha \in [-\frac{1}{c^{1/4}},\frac{1}{c^{1/4}}]} n_{\alpha} (\alpha \log n)^q}_{B} + \underbrace{\sum_{\substack{\alpha \in [-0.2, -\frac{1}{c^{1/4}}] \cup [\frac{1}{c^{1/4}}, 0.2]}}n_{\alpha}(\alpha \log n)^q}_{C}.
\]

For the left part, $B$, since $\alpha \in [-\frac{1}{c^{1/4}},\frac{1}{c^{1/4}}]$ implies that $(\alpha \log n)^q \leq (c^{-1/4}\log n)^q$, we get 
\begin{align*}
B & ~ \leq \left( \sum_{\alpha \in [-\frac{1}{c^{1/4}},\frac{1}{c^{1/4}}]} n_{\alpha} \right) \cdot (c^{-1/4}\log n)^q  \leq s \cdot (c^{-1/4}\log n)^q.
\end{align*}

For the right part, $C$, we use the bound on $n_{\alpha}$ from $E_3$,
\begin{align*}
C \leq & ~ \sum_{\substack{\alpha \in [-0.2, -\frac{1}{c^{1/4}}] \cup [\frac{1}{c^{1/4}}, 0.2]}}\log n \cdot (1 + s \cdot n^{-\eps\alpha^2}) \cdot (\alpha \log n)^q \\
\leq & ~ \sum_{\substack{\alpha \in [-0.2, 0.2]}}\log n \cdot (1 + s \cdot n^{-\eps\alpha^2}) \cdot (\alpha \log n)^q \\
\leq & ~ \log n \cdot \sum_{\alpha\in [-0.2,0.2]} (\alpha \log n)^q + \log n \cdot \sum_{\substack{\alpha \in [-0.2, 0.2]}}s \cdot n^{-\eps\alpha^2} \cdot (\alpha \log n)^q\\
\leq & ~ \log n \cdot \sum_{\alpha\in [-0.2,0.2]} (\alpha \log n)^q + s\cdot 0.4(\log n)^{q+2}\max_{\substack{\alpha \in [-0.2, 0.2]}} n^{-\eps\alpha^2}\alpha^q\\
\leq & ~ (0.2 \log n)^{q+2} + s\cdot (\log n)^{q+2} \cdot n^{-\Theta(\frac{\log \log c}{\log c})},
\end{align*}
where in the last step, we used that the quantity $n^{-\eps\alpha^2}\alpha^q$ is maximized at $\alpha = \Theta(\sqrt{\frac{\log \log c}{\log c}})$, in which case $\alpha^q = n^{-\eps\alpha^2} = n^{\Theta(\frac{\log \log c}{\log c})}$.

Thus,
\begin{align*}
A_{i,j} = B + C \leq & ~ (0.2\log n)^{q+2} + s(\log n)^{q+2}((c^{-1/4})^q + n^{-\Theta(\frac{\log \log c}{\log c})})\\
\leq & ~ (0.2\log n)^{q+2} + s(\log n)^{q+2}(n^{-\Theta(\frac{\log \log c}{\log c})}),
\end{align*}
where the last step is because $c^{-q/4} = n^{-\Omega(1)} \ll n^{-\Theta(\frac{\log \log c}{\log c})}$.


Finally, to ensure that $A_{i,j} \leq (p^2d)^q$, We will set $s$ such that $A_{i,j} \leq (\log n)^q \leq (p^2 d)^q$, since $p\geq 1/\sqrt{c}$ and thus $(p^2d)^q\geq (\log n)^q$. 
The $s$ we choose is
\begin{align*}
s
\leq & ~ \frac{(\log n)^q - (0.2 \log n)^{q+2}}{(\log n)^{q+2} (n^{-\Theta(\frac{\log \log c}{\log c})})} \\
= & ~ ((\log n)^{-2}-0.2^{q+2})n^{\Theta(\frac{\log \log c}{\log c})}\\
= & ~ n^{\Theta(\frac{\log \log c}{\log c})} .
\end{align*}

\section{Closest Pair under Random distribution}
\subsection{Preliminary}
In this section, we show how to solve the $\textbf{CP}_{n,d}$ problem (Definition~\ref{def:cp}) for dimension $d = c \log n$ in time $n^{2-\Omega(\frac{\log \log c}{\log c})}$, thus
proving Theorem~\ref{thm:main2}. The algorithm is nearly identical to the algorithm above for the Orthogonal Vectors problem, and our main work will be to prove a new version of Lemma~\ref{lem:core_lemma} in this setting.

We start with a standard manipulation so that we can focus on inner products rather than distances between vectors. Let $\mathbf{1}_d$ denote the all-1s vector of length $d$. Note that for two vectors $u,v \in \{0,1\}^d$, we always have $\|u-v\|_1 = \langle u, \mathbf{1}_d-v\rangle + \langle \mathbf{1}_d-u, v\rangle$. Thus, if we pad $ (\mathbf{1}_d-u)$ at the end of $u$ to get $x \in \{0,1\}^{2d}$, and pad $(\mathbf{1}_d-v)$ at the beginning of $v$ to get $y\in \{0,1\}^{2d}$, then we have $\langle x,y\rangle = \|u-v\|_1$.
In the remainder of this section, we will use letters $u,v$ to refer to the original vectors, and use letters $x\in X,y\in Y$ to refer to these padded versions that satisfy $\langle x,y\rangle = \|u-v\|_1$.

Let $\delta \in (0,1/2)$ denote the number such that the probability 
\begin{align}\label{eq:delta}
\Pr_{u,v\sim\{0,1\}^d}[\|u-v\|_1 \leq (1/2 - \delta)d] = 1/\sqrt{n}.
\end{align}
Note by a Chernoff bound that $\delta = \Theta(1/\sqrt{c})$. We see that in expectation, $n^{1.5}$ pairs of vectors will have distance $\leq (1/2-\delta)d$. We may thus assume the closest pair has distance $\leq (1/2-\delta)d$, since this is the case with probability $1-o(1)$.

Our algorithm will thus proceed by enumerating over all integers $t$ from $0$ to $(1/2-\delta)d$, and checking whether any pair has inner product $t$. This gives a negligible running time overhead of $O(d) \leq O(\log n)$.

\subsection{Our algorithm and analysis}

We will use almost the same algorithm as for Orthogonal Vectors, with the same choice of $q$ and $s$, i.e., $q=O(\frac{\log n}{\log c})$ is even, and $s=n^{\Theta(\frac{\log \log c}{\log c})}$. Here, we pick the polynomial to be $Q(x) := (x-d/2)^q$, since the expected inner product between a pair of randomly-drawn vectors is $d/2$.

\textbf{Algorithm.}
We partition arbitrarily $X$ into groups of the same size $s$. For every group $X_i$ and vector $y_j\in Y$, we calculate $A_{i,j} = \sum_{x\in X_i} Q(\langle x,y_j \rangle)$. If $A_{i,j} \geq (d/2-t)^q$, we brute force check if any $x\in X_i$ has inner product $t$ with $y_j$. 

The correctness of this algorithm follows again since, if some $x\in X_i$ has $\langle x,y_j\rangle =t$, then $A_{i,j}$ must be at least $(d/2-t)^q$ and we will find it in a brute force. The running time proof is also nearly identical, once we substitute in the following analogue of Lemma~\ref{lem:core_lemma}.

\subsection{Key Analysis Lemma}
Our new replacement for Lemma~\ref{lem:core_lemma} is as follows.
\begin{lemma}[Closest pair version]\label{lem:core_lemma_cp}
Let $q = \frac{\log n}{50\log c}$. There's a choice of the group size $s$, $s = n^{\Theta(\frac{\log \log c}{\log c})}$, such that the following is true. Consider a group $X_i$ and a vector $y_j$. With $1-n^{-1/200}$ probability over the randomness of $X_i$ and $y_j$, we have 
\[
A_{i,j} < (d/2-t)^q.
\]
\end{lemma}
\begin{proof}
Similar to Lemma~\ref{lem:core_lemma}, our proof will proceed by defining two random events, $E_2$ and $E_3$, then proving that
\begin{enumerate}[label=\textbf{(\arabic*)}]
    \item The probability that both $E_2$, $E_3$ happens is at least $1-n^{-1/200}$, and 
    \item If both of them happen, then $A_{i,j} < (d/2-t)^q$.
\end{enumerate}

We now define the events.
\begin{itemize}
\item Let $E_2$ be: for all $x\in X_i$, we have $\langle x,y_j\rangle \in (1/2 \pm \delta_2)d$. Here, $\delta_2$ is defined such that 
\begin{align}\label{eq:delta2}
\Pr_{u,v\sim \{0,1\}^d}[\|u-v\|_1 \leq (1/2-\delta_2)d]=1/n^{0.01}.
\end{align}
\item Let $E_3$ be the following. For all $\alpha$ in $[-\delta_2 \sqrt{c}, \delta_2\sqrt{c}]$ for which $(1/2 + \alpha/\sqrt{c})d$ is an integer, let $n_{\alpha}$ be the number of $x\in X_i$ that $\langle x,y_j\rangle = (1/2 + \alpha/\sqrt{c})d$. We define $E_3$ to be the case when
\[
n_{\alpha} \leq \log n \cdot (1 + s \cdot n^{-0.5\alpha^2}), ~~~~~~\forall \alpha \in [-\delta_2\sqrt{c}, \delta_2\sqrt{c}].
\]
\end{itemize}

\noindent\textbf{Proof of (1).}

\begin{itemize}
\item We have $\Pr[\neg E_2] \leq s\cdot n^{-0.01}$. This is by definition of $\delta_2$ and applying union bound on all $x\in X_i$.
    
\item We have $\Pr[\neg E_3]\leq 1/n^{10}$. 
To see this, note first that by a Chernoff bound, for a random $x$ from $X_i$, with $\mu = d/2$ and $\delta = \frac{\alpha}{2\sqrt{c}}$, we have 
\[
\Pr[|\langle x,y_j \rangle  - d/2| \geq \frac{\alpha}{\sqrt{c}}d]\leq n^{-0.5\alpha^2},~\forall \alpha > 0.
\]

For two vectors $x_1,x_2\in X_i$, one can see that by symmetry, the value $\langle x_1,y_j\rangle$ is independent of the value $\langle x_2,y_j\rangle$.  Thus, by a Chernoff bound argument like before, we have $n_{\alpha} \leq \log n \cdot (1 + s \cdot n^{-0.5\alpha^2})$ with probability $1-1/n^{10}$. 

\end{itemize}

Thus, by a union bound, we have
\[
\Pr[E_2,E_3]\geq 1- sn^{-0.01}-1/n^{10} \geq 1-n^{-1/200}.
\]

\noindent\textbf{Proof of (2).}

Recall that $A_{i,j} = \sum_{x\in X_i} Q(\langle x,y_j \rangle) = \sum_{x\in X_i} (\langle x,y_j \rangle - d/2)^q$. Supposing $E_2$ is true, we can rewrite $A_{i,j}$ similar to before as $A_{i,j} = \sum_{\alpha \in [-\delta_2\sqrt{c},\delta_2\sqrt{c}]} n_{\alpha} (\alpha \sqrt{c}\log n)^q$. 

Thus, $A_{i,j}$ is upper bounded by
\begin{align*}
\frac{A_{i,j}}{\sqrt{c}^q}
= & ~ \sum_{\alpha \in [-\delta_2\sqrt{c}, \delta_2\sqrt{c}]} n_{\alpha} (\alpha \log n)^q \\
= & ~ \log n \cdot \sum_{\alpha \in [-\delta_2\sqrt{c}, \delta_2\sqrt{c}]} (1+sn^{-0.5\alpha^2}) (\alpha \log n)^q \\
\leq & ~  \log n\sum_{\alpha\in [-\delta_2\sqrt{c}, \delta_2\sqrt{c}]} (\alpha \log n)^q +s\cdot(\log n)^{q+2}\max_{\substack{\alpha \in [-\delta_2\sqrt{c}, \delta_2\sqrt{c}]}} n^{-0.5\alpha^2}\alpha^q\\
\leq & ~  (\delta_2\sqrt{c}\log n)^{q+2} + s(\log n)^{q+2} n^{-\Theta(\frac{\log \log c}{\log c})},
\end{align*}
where in the last step, we used (as before) that the quantity $n^{-0.5\alpha^2}\alpha^q$ is maximized at $\alpha = \Theta(\sqrt{\frac{\log \log c}{\log c}})$, at which $\alpha^q  = n^{-0.5\alpha^2}= n^{-\Theta(\frac{\log \log c}{\log c})}$.

Note that $(\delta d)^q\leq (d/2-t)^q$ since $t \leq (1/2-\delta)d$. Our goal is to set $s$ such that $A_{i,j}\leq (\delta d)^q$, which implies as desired that $A_{i,j}\leq (d/2-t)^q$. 

Before setting $s$, we make one final definition. Note that by comparing the definitions of $\delta$ (Eq.~\eqref{eq:delta}) and $\delta_2$ (Eq.~\eqref{eq:delta2}), we have $\delta,\delta_2 = \Theta(1/\sqrt{c})$.  Define the constants $k, k_2 > 0$ such that $\delta = k/\sqrt{c}$ and $\delta_2 = k_2/\sqrt{c}$. By definition, the ratio $\frac{k}{k_2}$ is strictly greater than $1$. Thus, $A_{i,j}\leq (\delta d)^q$ means $A_{i,j}/\sqrt{c}^q \leq (\delta d)^q/\sqrt{c}^q =(k\log n)^q$. 

Thus, we choose $s$ to be
\begin{align*}
s 
\leq & ~ \frac{(k\log n)^q - (\delta_2 \sqrt{c} \log n)^{q+2}}{(\log n)^{q+2} (n^{-\Theta(\frac{\log \log c}{\log c})})} \\
= & ~ n^{\Theta(\frac{\log \log c}{\log c})}\cdot (k^q(\log n)^{-2} - (\delta_2 \sqrt{c})^{q+2})\\
= & ~ n^{\Theta(\frac{\log \log c}{\log c})}\cdot (k^q(\log n)^{-2} - k_2^{q+2})\\
= & ~ n^{\Theta(\frac{\log \log c}{\log c})}\cdot ((1-o(1))k^q(\log n)^{-2})\\
= & ~ n^{\Theta(\frac{\log \log c}{\log c})} \cdot k^{\frac{\log n}{50 \log c}} \\
= & ~ n^{\Theta(\frac{\log \log c}{\log c})},
\end{align*}
where the fourth step is because $(k/k_2)^q = (k/k_2)^{\Theta(\frac{\log n}{\log c})} \gg (\log n)^2$.
\end{proof}

\section{Conclusion}

In this paper, we gave faster algorithms for the average-case problems $\OV(p)_{n,d}$ and $\CP_{n,d}$ than the best-known worst-case algorithms. While these are not truly subquadratic-time algorithms, since the savings in the exponent depend on the constant $c$ such that $d = c \log n$, they nonetheless suggest that these problems may be easier in the average-case.

On the other hand, it seems difficult to use techniques similar to ours to get any further improvements. For instance, as in prior work~\cite{valiant2015finding,alman2016polynomial}, one may consider a more intricate polynomial such as a Chebyshev polynomial to replace our simple polynomial $P$, but one can verify that this could only decrease the needed degree $q$ by a constant factor: Chebyshev polynomials help most when one is distinguishing an orthogonal pair from many vectors with small inner product, whereas here we critically use that most pairs of vectors have large inner product.

We leave open the exciting questions of whether further improvements are possible for average-case OV, and whether one could use our techniques to improve the best algorithm for worst-case OV.

\bibliographystyle{alpha}
\bibliography{myref}
\end{document}

%% file: notations.tex
\newtheorem{theorem}{Theorem}[section]
\newtheorem{lemma}[theorem]{Lemma}
\newtheorem{definition}[theorem]{Definition}

\newtheorem{conjecture}[theorem]{Conjecture}

\newcommand{\wt}{\widetilde}

\newcommand{\eps}{\epsilon}

\renewcommand{\varepsilon}{\epsilon}
\renewcommand{\tilde}{\wt}

\renewcommand{\eps}{\epsilon}